\documentclass{llncs}

\usepackage{llncsdoc}
\usepackage{algorithm}
\usepackage{algorithmic}
\usepackage{graphicx}

\usepackage{amsmath}
\usepackage{amsfonts}
\usepackage{amssymb}
\usepackage{amssymb}
\usepackage{rotating}

\spnewtheorem{fact}{Fact}{\bfseries}{\itshape}
\spnewtheorem*{my_definition}{Definition}{\bfseries}{\itshape}

\spnewtheorem*{swap}{Problem ``SM"
 (Pattern Matching with Swaps)}{\bfseries}{\itshape}
\newcommand{\rom}[1]{\uppercase\expandafter{\romannumeral #1\relax}}

\begin{document}
\pagestyle{headings}

\pagestyle{headings}

\mainmatter

\title{The Swap Matching Problem Revisited}

\titlerunning{Lecture Notes in Computer Science}

\author{${\text{Pritom Ahmed}}^1$ \and ${\text{Costas S. Iliopoulos}}^2$ \and ${\text{A.S.M. Sohidull Islam}}^1$ \and ${\text{M. Sohel Rahman}}^1$}

\institute{A$\ell$EDA Group,\\
Department of Computer Science, BUET, Dhaka,\\
\email{\{pritom.11,sohansayed\}@gmail.com, msrahman@cse.buet.ac.bd}\\
\texttt{http://teacher.buet.ac.bd/msrahman}
\and 
Algorithm Design Group,\\
Department of Computer Science, King's College London, University of London\\
\email{csi@dcs.kcl.ac.uk}\\
http://www.dcs.kcl.ac.uk/staff/csi}

\maketitle

\begin{abstract}
In this paper, we revisit the much studied problem of Pattern
Matching with Swaps (Swap Matching problem, for short). We first
present a graph-theoretic model, which
opens a new and so far unexplored avenue to solve the problem. Then,
using the model, we devise two efficient algorithms to solve the swap
matching problem. The resulting algorithms are adaptations of the
classic shift-and algorithm. For patterns having length similar to
the word-size of the target machine, both the algorithms run in linear time considering a fixed alphabet.
\end{abstract}

\keywordname{ Algorithms; Strings; Swap Matching; Graphs.}

\section{\label{intro}Introduction}
The classical pattern matching problem is to find all the occurrences of a given pattern $P$ of length $m$ in a text $T$ of length $n$, both being sequences of characters drawn from a finite character set $\Sigma$. This problem is interesting as a fundamental computer science problem and is a basic need of many practical applications such as text retrieval, music information retrieval, computational biology, data mining, network security, among many others. In this paper, we revisit the Pattern Matching with Swaps problem (the Swap Matching problem, for short), which is a well-studied variant of the classic pattern matching problem. In this problem, the pattern $P$ is said to $swap~match$ the text $T$ at a given location $i$, if adjacent pattern characters can be swapped, if necessary, so as to make the pattern identical to the substring of the text ending (or equivalently, starting) at location $i$. All the swaps are constrained to be disjoint, i.e., each character is involved in at most one swap.

Amir et al.~\cite{DBLP:journals/jal/AmirALLL00} obtained the first non-trivial results for this problem. They showed how to solve the problem in time $O(nm^{1/3} \log m\log\sigma)$, where $\sigma = \min(|\Sigma|,m)$. Amir et al.~\cite{DBLP:journals/ipl/AmirLLL98} also studied certain special cases for which $O(n \log^2 m)$ time solution can be obtained. However, these cases are rather restrictive. 
Later, Amir et al.~\cite{DBLP:journals/iandc/AmirCHLP03} solved the Swap Matching problem in time $O(n\log m\log\sigma)$. We remark that all the above solutions to swap matching depend on Fast Fourier Transformation (FFT) technique. Recently, Cantone and Faro~\cite{CS} presented a dynammic programming approach to solve the swap matching problem which runs in linear time for finite character set $\Sigma$, when patterns are compatible with the word size of the target machine. Notably the work of~\cite{CS} avoids the use of FFT technique. Cantone, Faro and Campanelli presented another approach in~\cite{CCS} to solve the Swap matching problem. Though the algorithm of~\cite{CCS} runs in $O(nm)$ time for patterns compatible with the word size of the target machine, in practice it achieves quite good result. In fact as it turns out, the algorithm of~\cite{CCS} outperforms the algorithm of~\cite{CS} most of the time.  
Notably, approximate swapped matching~\cite{DBLP:journals/ipl/AmirLP02} and swap matching in weighted sequences~\cite{DBLP:conf/cis/ZhangGI04} have also been studied in the literature.

\subsection{Our Contribution}
The contribution of this paper is as follows. We first present a graph-theoretic approach to model the swap matching problem. Using this model, we devise two efficient algorithms to solve the swap matching problem. The resulting algorithms are adaptation of the classic shift-and algorithm~\cite{CharrasL04} and runs in linear time if the pattern size is similar to the size of word in the target machine, assuming a fixed alphabet size. Notably, some preliminary results of this paper were presented in~\cite{IR}. In~\cite{IR}, an algorithm running in $O(m/w n\log m)$ time was presented, where $w$ is the machine word size. For short patterns, i.e., pattern size similar to machine word size, this runtime becomes $O(n\log m)$. Hence the result in this paper clearly improves the results of~\cite{IR} and matches the result of~\cite{CS}. Finally, we present experimental results to compare the non-FFT algorithms of~\cite{CS,CCS} and our work.  

\subsection{RoadMap}
The rest of the paper is organized as follows. In Section~\ref{Pre}, we present some preliminary notations and definitions. In Section~\ref{Model}, we present our model to solve the swap matching problem. In Section~\ref{Algo}, we present two different algorithms to solve the swap matching problem. Section~\ref{experiment}, presents the experimental results. Finally, we briefly conclude in Section~\ref{conclusion}.

\section{\label{Pre}Preliminaries}
A \emph{string} is a sequence of zero or more symbols from an alphabet, $\Sigma$. A string $X$ of length $n$ is denoted by $X[1..n] = X_1 X_2 \ldots X_{n}$, where $X_i \in \Sigma$ for $1 \leq i \leq n$. The \emph{length} of $X$ is denoted by $|X|= n$. A string $w$ is called a \emph{factor} of $X$ if $X = uwv$ for $u, v\in \Sigma^*$; in this case, the string $w$ occurs at position $|u|+1$ in $X$. The factor $w$ is denoted by $X[|u|+1..|u|+|w|]$. A \emph{$k$-factor} is a factor of length $k$. A \emph{prefix (or suffix)} of $X$ is a factor $X[x..y]$ such that $x=1~(y=n)$, $1\leq y \leq n~(1\leq x \leq n)$. We define the $i$-th prefix to be the prefix ending at position $i$, i.e., $X[1..i], 1\leq i\leq n$. On the other hand, the $i$-th suffix is the suffix starting at position $i$, i.e., $X[i..n], 1\leq i\leq n$.

\begin{definition}
A \emph{swap permutation} for $X$ is a permutation
$\pi: \{1,\ldots,n\}\rightarrow \{1,\ldots,n\}$ such that:
\begin{enumerate}
\item if $\pi(i) = j~ \text{then}~\pi(j) = i$ (characters are swapped).
\item for all $i, \pi(i)\in \{i-1,i,i+1\}$ (only adjacent characters are swapped).
\item if $\pi(i) \neq i ~\text{then}~ X_{\pi(i)} \neq X_i$ (identical characters are not
swapped).
\end{enumerate}
\end{definition}

For a given string $X$ and a swap permutation $\pi$ for $X$, we use
$\pi(X)$ to denote the \emph{swapped version} of $X$, where
$\pi(X)=X_{\pi(1)} X_{\pi(2)} \ldots X_{\pi(n)}$.

\begin{definition}
Given a text $T = T_1 T_2 \ldots T_n$ and a pattern $P = P_1 P_2
\ldots P_m$, $P$ is said to swap match at location $i$ of $T$
if there exists a swapped version $P'$ of $P$ that matches $T$ at
location\footnote{Note that, we are using the end position of the
match to identify it.} $i$, i.e. $P'_j = T_{i-m +j}$ for
$j\in[1..m]$.
\end{definition}

\begin{swap}
Given a text $T = T_1 T_2 \ldots T_n$ and a pattern $P = P_1 P_2
\ldots P_m$, we want to find each location $i\in[1..n]$ such that
$P$ swap matches with $T$ at location $i$.
\end{swap}

%
\begin{definition}\label{Def_Degenerate_String}
A string $X$ is said to be degenerate, if it is built over the
potential $2^{|\Sigma|}-1$ non-empty sets of letters belonging to
$\Sigma$.
\end{definition}

\begin{example}
Suppose we are considering DNA alphabet, i.e., $\Sigma = \Sigma_{DNA}=\{A,C,T,G\}$. Then we have 15 non-empty sets of letters belonging to $\Sigma_{DNA}$. In what follows, the set containing $A$ and $T$ will be denoted by $[AT]$ and the singleton $[C]$ will be simply denoted by $C$ for ease of reading. The set containing all the letters, namely $[ACTG]$, is known as the don't care character in the literature.
\end{example}

\begin{definition}\label{Def_Degenerate_Match}
Given two degenerate strings $X$ and $Y$ each of length $n$, we say
$X[i]$ \emph{matches} $Y[j], 1\leq i,j\leq n$ 
if, and only if, $X[i]\cap Y[j]\neq \emptyset$.
\end{definition}

\begin{example}\label{Ex_Deg_Match}
Suppose we have degenerate strings $X = AC[CTG]TG[AC]C$ and $Y =
TC[AT][AT]TTC$. Here, $X[3]$ matches $Y[3]$ because $X[3] = [CTG]
\cap Y[3] = [AT] = T \neq \emptyset$. 
\end{example}

\section{\label{Model}The Graph-Theoretic Model for Swap Matching}
In this section, we propose the graph-theoretic model to solve the swap matching problem. In this model, both the text and the pattern are viewed as two separate graphs. We start with the following definitions.

\begin{definition}
The \textbf{\emph{$\mathcal T$-graph}} is defined in the following way:

Given a text $T=T_1\dots T_n$ of Problem SM, a
\textbf{\emph{$\mathcal T$-graph}}, denoted by $T^G = (V^T, E^T) $,
is a directed graph with $n$ vertices and $n-1$ edges such that $V^T
= \{1, 2, \ldots n\}$ and $E^T = \{(i,i+1)~|~1\leq i<n\}$. For each
$i\in V^T$, we define $label(i) = T_i$ and for each edge $e \equiv
(i,j) \in E_T$, we define $label(e) \equiv label((i,j)) \equiv
(label(i), label(j)) = (T_i,T_j)$.
\end{definition}

Note that, the labels in the above definition may not be unique.
Also, we normally use the labels of the vertices and the edges to
refer to them.

\begin{figure}[h!]
\begin{center}

\begin{tabular}{ccccccccccccccccccccccccccccc}
a&$\rightarrow$&c&$\rightarrow$&a&$\rightarrow$&c&$\rightarrow$&b&$\rightarrow$&a&$\rightarrow$&c&$\rightarrow$&c&$\rightarrow$&b&$\rightarrow$&a&$\rightarrow$&c&$\rightarrow$&a&$\rightarrow$&c&$\rightarrow$&b&$\rightarrow$&a\\
\end{tabular}
 \caption{The corresponding $\mathcal T$-graph of Example~\ref{Ex_T_Graph}}
\label{Fig_T_Graph}
\end{center}
\end{figure}

\begin{example}\label{Ex_T_Graph}
Suppose, $T = acacbaccbacacba$. Then the corresponding $T$-graph is
shown in Figure~\ref{Fig_T_Graph}.
\end{example}

\begin{definition}

The \textbf{\emph{$\mathcal P$-graph}} is defined in the following way:

Given a text $P=P_1\dots P_m$ of Problem SM, a
\textbf{\emph{$\mathcal P$-graph}}, denoted by $P^G = (V^P, E^P) $,
is a directed graph with $3m-2$ vertices and at most $5m-9$ edges.
The vertex set $V^P$ can be partitioned into three disjoint vertex
sets, namely, $V^P_{(+1)},V^P_{0},V^P_{(-1)}$ such that $|V^P_{(+1)}|
= |V^P_{(-1)}| = m -1 ~\text{and}~ |V^P_{(0)}| = m$. The partition
is defined in a $3\times m$ matrix $M[3,m]$ as follows. For the sake
of notational symmetry we use $M[-1], M[0]$ and $M[+1]$ to denote
respectively the rows $M[1], M[2]$ and $M[3]$ of the matrix $M$.
\begin{enumerate}
\item $V^P_{(-1)} = \{M[-1,2], M[-1,3], \ldots  M[-1,m]\}$
\item $V^P_{(0)} = \{M[0,1], M[0,2], \ldots  M[0,m]\}$
\item $V^P_{(+1)} = \{M[+1,1], M[+1,2], \ldots  M[+1,m-1]\}$
\end{enumerate}
The labels of the vertices are derived from $P$ as follows:

\begin{enumerate}
\item For each vertex $M[-1,i]\in V^P_{(-1)}, 1<i\leq m$,
label(M[-1,i]) =  $P_{i-1}$

\item For each vertex $M[0,i]\in V^P_{(0)}, 1\leq i\leq m, label(M[0,i]) =
P_i$
\item For each vertex $M[+1,i]\in V^P_{(+1)}, 1\leq i< m$,
label(M[+1,i]) = $P_{i+1}$

\end{enumerate}

The edge set $E^P$ is defined as the union of the sets $E^P_{(-1)},
E^P_{(0)}$ and $E^P_{(+1)}$ as follows:

\begin{enumerate}
\item $E^P_{(-1)} = \{(M[-1,i], M[0,i+1]), (M[-1,i], M[+1,i+1])~|~ 2\leq i\leq m-2\}$
\item $E^P_{(0)} = \{(M[0,i], M[0,i+1])~|~ 1\leq i\leq m-1\} ~\bigcup~ \{((M[0,i], M[+1,i+1])~|~ 1\leq i\leq m-2\}$
\item $E^P_{(+1)} = \{(M[+1,i], M[-1,i+1])~|~ 1\leq i\leq m-1\}$
\end{enumerate}

The labels of the edges are derived from using the labels of the
vertices in the obvious way.
\end{definition}

\begin{figure}[hbtp]
\begin{center}
\includegraphics[width=0.8\textwidth,angle=0]{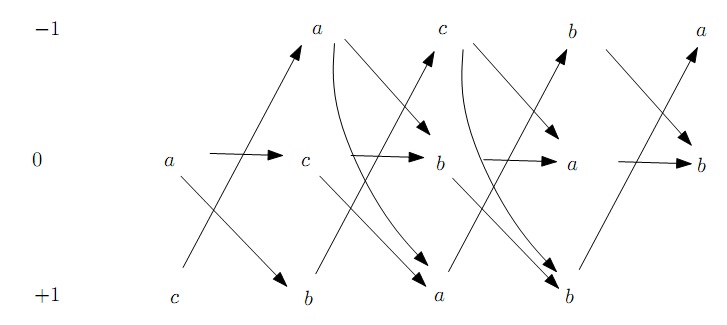}
\caption{$\mathcal P$-graph of the Pattern P = acbab} \label{Fig_P_Graph1}
\end{center}
\end{figure}

\begin{example}\label{Ex_P_Graph}
Suppose, $P = acbab$. Then the corresponding $\mathcal P$-graph
$P^G$ is shown in Figure~\ref{Fig_P_Graph1}.
\end{example}

\begin{definition}

Given a $\mathcal P$-graph $P^G$, a path $Q = u_1\rightsquigarrow
u_\ell = u_1 u_2 \ldots u_\ell $ is a sequence of consecutive
directed edges $\langle(u_1,u_2),(u_2,u_3),\dots(u_{\ell-1},u_\ell)
\rangle$ in $P^G$ starting at node $u_1$ and ending at node
$u_\ell$. The length of the path $Q$, denoted by $len(Q)$, is the
number of edges on the path and hence is $\ell -1$ in this case. It
is easy to note that the length of a longest path in $P^G$ is $m-1$.
\end{definition}

\begin{definition}

Given a $\mathcal P$-graph $P^G$ and a $\mathcal T$-graph $T^G$, we
say that $P^G$ matches $T^G$ at position $i\in[1..n]$ if, and only if,
there exists a path $Q = u_1 u_2 \ldots u_m$ in $P^G$ having
$u_1\in\{M[0,1], M[+1,1]\}$ and $u_m\in\{M[-1,m], M[0,m]\}$ such that
for $j\in[1..m]$ we have $label(u_j) = T_{i-m+j}$.
\end{definition}

This completes the definition of the graph theoretic model. The following Lemma presents the idea to solve the swap matching problem using the presented model.

\begin{lemma}
Given a pattern $P$ of length $m$ and a text $T$ of length $n$,
suppose $P^G$ and $T^G$ are the $\mathcal P$-graph and $\mathcal
T$-graph of $P$ and $T$, respectively. Then, $P$ swap matches $T$ at
location $i\in [1..n]$ of $T$ if and only if $P^G$ matches $T^G$ at
position $i\in [1..n]$ of $T^G$.
\end{lemma}

\begin{proof}
The proof basically follows easily from the definition of the
$\mathcal P$-graph. At each column of the matrix $M$, we have all
the characters as nodes considering the possible swaps as explained
below. Each node in row $(-1)$ and $(+1)$ represents a swapped
situation. Now consider column $i$ of $M$ corresponding to $P^G$.
According to definition, we have $M[-1,i] = P_{i-1}$ and $M[+1,i-1]
= P_{i}$. These two nodes represents the swap of $P_{i}$ and
$P_{i-1}$. Now, if this swap takes place, then in the resulting
pattern, $P_{i-1}$ must be followed by $P_i$. To ensure that, in
$P^G$, the only edge starting at $M[+1,i-1]$, goes to $M[-1,i]$. On
the other hand, from $M[-1,i]$ we can either go to $M[0,i+1]$ or to
$M[+1,i+1]$: the former is when there is no swap for the next pair
and the later is when there is another swap for the next pair.
Recall that, according to the definition, the swaps are disjoint.
Finally, the nodes in row $0$ represents the normal (non-swapped)
situation. As a result, from each $M[0,i]$ we have an edge to
$M[0,i+1]$ and an edge to $M[+1,i+1]$: the former is when there is
no swap for the next pair as well and the later is when there is a
swap for the next pair. So it is easy to see that all the paths of
length $m-1$ in $P^G$ represents all combinations considering all
possible swaps in $P$. Hence the result follows.~$\Box$
\end{proof}

Since the number of the possible paths of length $m-1$ in $P^G$ is exponential in $m$, we exploit the above model in a different way and apply a modified version of the classic shift-and~\cite{BG92} algorithm to solve the swap matching problem.

\section{\label{Algo}Our Algorithms for Swap Matching}
In this section, we use the model proposed in Section~\ref{Model} to devise two novel efficient algorithms for the swap matching problem. Both of the algorithms are modified versions of the classic shift-and algorithm for pattern matching. We start with a brief review of the shift-and algorithm below. In Sections~\ref{AIIR} and~\ref{NewAlgo} we present the modifications needed to adapt it to solve the swap matching problem.

\subsection{\label{shift-and}Shift-And Algorithm}
The shift-and algorithm uses the bitwise techniques and is very
efficient if the size of the pattern is no greater than the word
size of the target processor. The following description of the
shift-and algorithm is taken from~\cite{CharrasL04} after slight
adaptation to accommodate our notations.

Let $R$ be a bit array of size $m$. Vector $R_j$ is the value of the
array $R$ after text character $T_j$ has been processed. It contains
information about all matches of prefixes of $P$ that end at
position $j$ in the text. So, for $1\leq i\leq m$ we have:
\begin{equation}
R_j[i] = \begin{cases}
                            1& \text{if $P[1..i] = T[j-i+1..j]$},\\
                            0& \text{Otherwise.}
         \end{cases}
\end{equation}

The vector $R_{j+1}$ can be computed after $R_j$ as follows. For
each $R_j[i]=0$:
\begin{equation}
R_{j+1}[i+1] = \begin{cases}
                            1& \text{if $P_{i+1} = T_{j+1}$},\\
                            0& \text{Otherwise.}
         \end{cases}
\end{equation}

and

\begin{equation}
R_{j+1}[0] = \begin{cases}
                            1& \text{if $P_0 = T_{j+1}$},\\
                            0& \text{Otherwise.}
         \end{cases}
\end{equation}

If $R_{j+1}[m]=1$ then a complete match can be reported. The transition from $R_j$ to $R_{j+1}$ can be computed very fast as
follows. For each $c\in\Sigma$ let $D_c$ be a bit array of size $m$
such that for $1\leq i\leq m$, $D_c[i]=1$ if and only if $P_i=c$.The
array $D_c$ denotes the positions of the character $c$ in the
pattern $P$. Each $D_c$ for all $c\in\Sigma$ can be preprocessed
before the pattern search. Then the computation of $R_{j+1}$ reduces
to two simple operations, namely, \textbf{shift} and \textbf{and} as follows: $$R_{j+1}=~ SHIFT(R_j)~ AND~
D_{T_{j+1}}$$

\subsection{The First Algorithm: SMALGO-\rom{1}}{\label{AIIR}

In this section, we present a modification of  the shift-and algorithm to solve the swap matching problem using the graph model presented in Section~\ref{Model}. In what follows the resulting algorithm shall be referenced to as SMALGO-\rom{1}. The idea of SMALGO-\rom{1} is described below.

First of all, the shift-and algorithm can be extended easily for the degenerate patterns~\cite{BG92}. In our swap matching model the pattern can be thought of having a set of letters at each position as follows: $\tilde{P} = [M[0,1]M[+1,1]]~[M[-1,2]M[0,2]M[+1,2]]\ldots[M[-1,m-1]M[0,m-1][+1,m-1]]~[M[-1,m]M[0,m]]$. Note that we have used $\widetilde{P}$ instead of $P$ above because, in our case, the sets of characters in the consecutive positions in the pattern $P$ don't have the same relation as in a usual degenerate pattern. In particular, in our case, a match at position $i+1$ of $P$ will depend on the previous match of position $i$ as the following example shows.

\begin{example}\label{Ex_Diff}
Suppose, $P = acbab$ and $T = bcbaaabcba$. The $\mathcal P$-graph of
$P$ is shown in Figure~\ref{Fig_P_Graph1}. So, in line of above
discussion, we can say that $\widetilde{P} =
[ac][acb][cba][ba][ab]$. Now, as can be easily seen, if we consider
degenerate match, then $\widetilde{P}$ matches $T$ at Positions $2$
and $6$. However, $P$ swap matches $T$ only at Position $6$; not at
Position $2$. To elaborate, note that at Position $2$, the match is
due to `$c$'. So, according to the graph $P^G$ the next match has to
be an `$a$' and hence at Position 2 we can't have a swap match.
\end{example}

For the sake of convenience, in the discussion that
follows, we refer to both $\widetilde{P}$ and the pattern $P$ as
though they were equivalent; but it will be clear from the context
what we really mean. Suppose we have a match up to position $i< m$
of $\widetilde{P}$ in $T[j-i+1..j]$. Now we have to check whether
there is a match between $T_{j+1}$ and $P_{i+1}$. For simple
degenerate match, we only need to check whether $T_{j+1}\in P_{i+1}$
or not. However, as Example~\ref{Ex_Diff} shows, for our case we
need to do more than that.

In what follows, we present a novel technique to adapt the shift-and algorithm to tackle the above situation. Suppose that $T_j = c = M[\ell,i]$. Now, from the $\mathcal P$-graph we know which
of the $M[k,i+1], k\in[-1,0,+1]$ will follow $M[\ell,i]$ and which
of the $M[k,i+2], k\in[-1,0,+1]$ will follow $M[q,i+1]$. So, for
example, even if $M[q,i+1] = T[j+1]$ we can't continue if there is 
no edge from $M[\ell,i]$ to $M[q,i+1]$ or from $M[q,i+1]$ to $M[r,i+2]$ in the $\mathcal P$-graph.

To tackle this, we define a new notion. Consider 3-member vertex sets $\{u_0,u_1,u_2\}$ and $\{x_0,x_1,x_2\}$ of $\mathcal P$-graph such that there exist edges $(u_{i},u_{i+1})$ and $(x_{i},x_{i+1})$, for all $i$ where $0 \leq i < 2$. Then the edge $(u_0,u_1)$ and $(x_0,x_1)$ are considered to be $same$ if and only if, $label(u_{i})=label(x_{i})$ for all $i$ where $0 \leq i < 2$.

Also, given an edge $(u_0,u_1) \equiv (M[i_1,j_1],M[i_2,j_2])$, we say
that edge $(u_0,u_1)$ `belongs to' column $j_2$, i.e., where the edge
ends; and we say $col((u_0,u_1)) \equiv col((M[i_1,j_1],M[i_2,j_2])) =
j_2$. Now we traverse all the edges and construct a set of sets
$\mathcal S = \{S_1\ldots S_\ell\}$ such that each $S_i, 1\leq i\leq
\ell$ contains the edges that are `same'. The set $S_i$ is named by
$\{u_0,u_1,u_2\}$ and we may refer to $S_i$
using its name. Now, we construct $P$-masks $P_{S_j}, 1\leq j\leq \ell$ such that $P_{S_j}[g] = 1$ if and only if, there is a set of $three$ vertices $\{u_0,u_1,u_2\}$ such that there exists an edge between each $(u_{i},u_{i+1})$ where $0 \leq i \leq 1$ and $\{u_0,u_1,u_2\}\in S_{j}$ having $col((u_0,u_1)) = g$. With the $P$-masks at our hand, we compute $R_{j+1}$ as follows:

\begin{equation}\label{Rj+1}
R_{j+1}=~ RSHIFT(R_j)~ AND~ D_{T_{j+1}} ~ AND~ LSHIFT(D_{T_{j+2}})~ AND~ P_{(T_j,T_{j+1},T_{j+2})}
\end{equation}

Here, RSHIFT indicates right shift, LSHIFT indicates left shift and AND is the usual bitwise AND operation. Note that, to locate the appropriate $P$-mask, we again need to perform a look up in the database constructed during the construction of the $P$-masks. Since a particular $P$-mask involves a set of $3$ (consecutive) vertices, we need a $3$D array to ensure constant time reference to it. Note that in Equation~\ref{Rj+1}, we have referred to the $P$-mask using the $3$ vertices of the corresponding vertex set. Example~\ref{Ex_match1} presents a complete execution of our algorithm.

\begin{example}\label{Ex_match1}
Suppose, $P = acbab$ and $T = acbbabcabab$. The $D$-masks of
$P$ are shown in Table~\ref{common Dmask}. The $P$-masks of
$P$ are shown in Table~\ref{oldP_mask}. Table~\ref{calculation} shows the detail computation of $R$. Explanation of the terms used in the Table~\ref{calculation} are as follows:   

\begin{itemize}
\item[$\mathcal{R}$] Right Shift Operation on the previous column
\item[$D_x$] $D$-Mask value for character `x'
\item[$\mathcal{L}$ $D_x$] Left Shift Operation on $D_x$
\item[$P_{(x,y,z)}$] $P$-Mask value of the set (x,y,z)
\item[$R_j$] Value of $R$ after $T_j$ has processed ( `1' in $m$-$1$ th row of $R_j$ column indicates that a match has been found ending at the corresponding column.)
\end{itemize}
\end{example}
The preprocessing is formally presented in Algorithm~\ref{algorithm for preprocessing1}. The main algorithm is presented in Algorithm~\ref{Algorithm for Swap matching(Short and Improved)}.

\begin{table}
\begin{center}
\begin{tabular}{|c|c|c|c|c|c|}
\hline
~ ~ ~ ~ &~ ~ D~ ~ &~ $D_a$~ &~ $D_b$~ &~ $D_c$~ &~ $D_x$~ \\
\hline
1&[ac]&1&0&1&0\\
\hline
2&[acb]&1&1&1&0\\
\hline
3&[cba]&1&1&1&0\\
\hline
4&[bab]&1&1&0&0\\
\hline
5&[ab]&1&1&0&0\\
\hline
\end{tabular}
\caption{The $D$-masks for pattern in Examples~\ref{Ex_match1} and~\ref{Ex_match2} }
\label{common Dmask}
\end{center}
\end{table}

\begin{table}
\begin{center}
\begin{minipage}{\textwidth}
\begin{tabular}{|c|c|c|c|c|c|c|c|c|c|c|c|c|c|c|c|}
\hline
&$P_{(a,a,b)}$&$P_{(a,b,a)}$&$P_{(a,b,b)}$&$P_{(a,b,c)}$&$P_{(a,c,a)}$&$P_{(a,c,b)}$&$P_{(b,a,b)}$&$P_{(b,b,a)}$&$P_{(b,c,a)}$&$P_{(b,c,b)}$&$P_{(c,a,a)}$&$P_{(c,a,b)}$&$P_{(c,b,a)}$&$P_{(c,b,b)}$&$P_{(x,x,x)}$\footnote{Here, $P_{(x,x,x)}$ indicates the edges that are not present in the $\mathcal{P}$-graph.}\\
\hline
1&1&1&1&1&1&1&1&1&1&1&1&1&1&1&1\\
\hline
2&0&0&0&1&1&1&0&0&0&0&1&1&0&0&0\\
\hline
3&1&1&1&0&0&0&0&0&1&1&0&1&1&1&0\\
\hline
4&0&0&1&0&0&0&1&1&0&0&0&1&1&0&0\\
\hline
5&0&0&0&0&0&0&0&0&0&0&0&0&0&0&0\\
\hline
\end{tabular}
\end{minipage}
\caption{$P$-masks for pattern in Example~\ref{Ex_match1} for SMALGO-\rom{1}}
\label{oldP_mask}
\end{center}
\end{table}

\begin{sidewaystable}
\tiny 
\tabcolsep 0.1pt
\begin{minipage}{\textwidth}
\begin{tabular}{|c|c|c|c|c|c|c|c|c|c|c|c|c|c|c|c|c|c|c|c|c|c|c|c|c|c|c|c|c|c|c|c|c|c|c|c|c|c|c|c|c|c|c|c|c|c|c|c|c|c|c|}
\hline
&-&$\mathcal{R}$&$D_a$& $P_{(x,x,x)}$\footnote{Here, $P_{(x,x,x)}$ indicates the edges that are not present in the $\mathcal{P}$-graph.}&$R_1$   &$\mathcal{R}$&$D_c$&$\mathcal{L}$ $D_b$&$P_{(a,c,b)}$&$R_2$&$\mathcal{R}$&$D_b$&$\mathcal{L}$ $D_b$&$P_{(c,b,b)}$&$R_3$&$\mathcal{R}$&$D_b$&$\mathcal{L}$ $D_a$&$P_{(b,b,a)}$&$R_4$ 
&$\mathcal{R}$&$D_a$&$\mathcal{L}$ $D_b$&$P_{(b,a,b)}$&$R_5$&$\mathcal{R}$&$D_b$&$\mathcal{L}$ $D_c$&$P_{(a,b,c)}$&$R_6$&$\mathcal{R}$&$D_c$&$\mathcal{L}$ $D_a$&$P_{(b,c,a)}$&$R_7$&$\mathcal{R}$&$D_a$&$\mathcal{L}$ $D_b$&$P_{(c,a,b)}$&$R_8$&$\mathcal{R}$&$D_b$&$\mathcal{L}$ $D_a$&$P_{(a,b,a)}$&$R_9$&$\mathcal{R}$&$D_a$&$\mathcal{L}$ $D_b$&$P_{(b,a,b)}$&$R_{10}$ \\
\hline
1&0&1&1&1&1 &1&1&1&1&1  &1&0&1&1&0  &1&0&1&1&0  &1&1&1&1&1  &1&0&1&1&0  &1&1&1&1&1  &1&1&1&1&1  &1&0&1&1&0  &1&1&1&1&1\\
\hline
2&0&0&1&0&0 &1&1&1&1&1  &1&1&1&0&0  &0&1&1&0&0  &0&1&1&0&0  &1&1&1&1&1  &0&1&1&0&0  &1&1&1&1&1  &1&1&1&0&0  &0&1&1&0&0\\
\hline
3&0&0&1&0&0 &0&1&1&0&0  &1&1&1&1&1  &0&1&1&0&0  &0&1&1&0&0  &0&1&0&0&0  &1&1&1&1&1  &0&1&1&1&0  &1&1&1&1&1  &0&1&1&0&0\\
\hline
4&0&0&1&0&0 &0&0&1&0&0  &0&1&1&0&0  &1&1&1&1&1  &0&1&1&1&0  &0&1&0&0&0  &0&0&1&0&0  &1&1&1&1&1  &0&1&1&0&0  &1&1&1&1&1\\
\hline 
5&0&0&1&0&0 &0&0&0&0&0  &0&1&0&0&0  &0&1&0&0&0  &1&1&0&0&0  &0&1&0&0&0  &0&0&0&0&0  &0&1&0&0&0  &1&1&0&0&0  &0&1&0&0&0\\
\hline
\end{tabular}
\end{minipage}
\caption{Detailed Calculation for text in Example~\ref{Ex_match1} for SMALGO-\rom{1}} \label{calculation}
\end{sidewaystable}

\begin{algorithm}
\caption{Computation of Preprocessing [$P$-masks] for SMALGO-\rom{1}}
\label{algorithm for preprocessing1}
\begin{algorithmic}[1]
\REQUIRE{\textbf{Pattern p}}
\STATE{$x \leftarrow pattern_{length}$}
\FOR{$i=0$ to $x-2$}\STATE{$pmask[p[i]][p[i+1]][p[i+2]] \leftarrow pmask[p[i]][p[i+1]][p[i+2]] | (1<<(x-i-2))$}
\STATE{{$pmask[p[i]][p[i+2]][p[i+1]] \leftarrow pmask[p[i]][p[i+2]][p[i+1]] | (1<<(x-i-2))$}}
\ENDFOR
\FOR{$i=0$ to $x-3$}\STATE{$pmask[p[i]][p[i+1]][p[i+3]] \leftarrow pmask[p[i]][p[i+1]][p[i+3]] | (1<<(x-i-2)) $}
\ENDFOR
\FOR{$i=1$ to $x-1$}\STATE{$pmask[p[i]][p[i-1]][p[i+1]] \leftarrow pmask[p[i]][p[i-1]][p[i+1]] | (1<<(x-i-1))$}
\ENDFOR
\FOR{$i=1$ to $x-2$}\STATE{$pmask[p[i]][p[i-1]][p[i+2]] \leftarrow pmask[p[i]][p[i-1]][p[i+2]] | (1<<(x-i-1))$}
\ENDFOR
\FOR{$i=0$ to $x-3$}\STATE{$pmask[p[i]][p[i+3]][p[i+2]] \leftarrow pmask[p[i]][p[i+3]][p[i+2]] | (1<<(x-i-3))$}
\ENDFOR
\FOR{$i=0$ to $x-4$}\STATE{$pmask[p[i]][p[i+2]][p[i+4]] \leftarrow pmask[p[i]][p[i+2]][p[i+4]] | (1<<(x-i-3))$}
\ENDFOR
\FOR{$i=0$ to $x-3$}\STATE{$pmask[p[i]][p[i+2]][p[i+3]] \leftarrow pmask[p[i]][p[i+2]][p[i+3]] | (1<<(x-i-3))$}
\ENDFOR
\RETURN{\textbf{$P$-mask pmask[]} for Pattern p}
\end{algorithmic}
\end{algorithm}

\begin{algorithm}
\caption{SMALGO-\rom{1}}
\label{Algorithm for Swap matching(Short and Improved)}
\begin{algorithmic}[1]
\STATE{$R_{0} \leftarrow 2^{pattern_{length}}-1$}
\STATE{$R_{0} \leftarrow R_{0}$ \& $D_{T_{0}}$}
\STATE{$R_{1} \leftarrow R_{0} >> 1$}
\STATE{$x \leftarrow 2$}
\STATE{$i:=0$}
\FOR{$j=0$ to $(n-3)$}
	\STATE{find $D_{j}$ for $Text[j]$}
	\STATE{$R_{j} \leftarrow R_{j}$ \& $pmask_{(T_{j},T_{j+1},T_{j+2})}$ \& $D_{T_{j+1}}$\& $(D_{T_{j+2}} << 1)$}
	\IF{$(R_{j}$ \& $x) = x$} \STATE{Match found ending at position $(j-1)$}
	\ENDIF
	\STATE{$R_{j+1} \leftarrow R_{j} >> 1$}
\ENDFOR
\end{algorithmic}
\end{algorithm}

\subsection{Analysis of SMALGO-\rom{1}}

\begin{table}
\begin{center}
\begin{tabular}{|c|c|}
\hline
Phase&Running Time\\
\hline
Computation of $D$-masks&$O(m/w(m + |\Sigma|))$\\
\hline
Computation of $P$-masks&$O(m/w~(m + {\vert\Sigma\vert}^{3}))$\\
\hline
Running time of Algorithm~\ref{Algorithm for Swap matching(Short and Improved)}&$O(m/w ~n)$\\
\hline
\end{tabular}
\caption{Running times of the different phases of SMALGO-\rom{1}}
\label{Table_RunningTime1}
\end{center}
\end{table}

The running times of the different phases of SMALGO-\rom{1} is listed in Table~\ref{Table_RunningTime1}. In the algorithm, we first initialize all the entries of $P$-masks which requires $O(m/w {\vert\Sigma\vert}^{3})$ time. Then, we start traversing the edges and corresponding $P$-masks in a name database ($3$-D array). Finding and updating the $P$-mask of corresponding edges can be done in constant time. As we have $O(m)$ edges, the total time needed for the computation of $P$-masks is $O(m/w (m + {\vert\Sigma\vert}^{3}))$. 
The computation of $D$-masks takes $O(m/w (m + {\vert\Sigma\vert}))$ time~\cite{BG92} when pattern is not degenerate. However, in our case, we need to assume that our pattern has a set of letters in each position. In this case, we require $O(m/w(m'+\Sigma))$ time where $m'$ is the sum of the cardinality of the sets at each position~\cite{BG92}. In general degenerate strings, $m'$ can be $m|\Sigma|$ in the worst case. However, in our case, $m' = |V^P| = O(m)$, where $V^P$ is the vertex set of the $\mathcal P$-graph. So, computation of the $D$-mask requires $O(m/w(m+\Sigma))$ time in the worst case. So the whole preprocessing takes $O(m/w (m + {\vert\Sigma\vert}^{3} + m + {\vert\Sigma\vert})))$ time. Assuming constant alphabet $\Sigma$ and pattern of size compatible with machine word length the preprocessing time becomes $O(m)$.

With the $P$-masks and $D$-masks at our hand, for our problem, we simply need to compute $R_{j}$ using Equation~\ref{Rj+1}. So, in total the construction of $R$ values require $O(m/w ~n)$ which is $O(n)$ if $m \sim w$. Therefore, in total the running time for SMALGO-\rom{1}, is linear 
assuming a constant alphabet and a pattern size similar to the word size of the target machine.

\subsection{\label{NewAlgo}The Second Algorithm : SMALGO-\rom{2}} 

In this section, we present another algorithm which is more space efficient. Instead of a $3$-D array we need only $2$-D arrays here. In order to understand the new algorithm, we need the following definitions.

\begin{definition}\label{levelchange}
A \emph{level change} indicates a change of row in the Matrix M having one of the following cases :
\begin{itemize}
\item A \textbf{Upward Change}, i.e., going from a position $(+1,j)$ to $(-1,j+1)$; 

\item A \textbf{Downward Change}, i.e., going from a position $(i,j)$ to $(+1,j+1)$ where 

$i = 0$ or $i = -1$; 

\item A \textbf{End of Swap/Middle-ward Change}, i.e., going from a position $(-1,j)$ to $(0,j+1)$. 
\end{itemize}

\end{definition}
In this approach, we only need to know which
of the $M[k,i+1], k\in[-1,0,+1]$ will follow $M[\ell,i]$ in the $\mathcal P$-graph. Thus we have to generate $P$-masks in the following way. Here we change the notion of two edges being `same' as follows.

Two edges $(u,v),(x,y)$ of the $\mathcal
P$-graph are said to be `same' if $label(u) =label(x)$ and
$label(v) =label(y)$, i.e., if the two edges have the same labels.
Also, given an edge $(u,v) \equiv (M[i_1,j_1],M[i_2,j_2])$, we say
that edge $(u,v)$ `belongs to' column $j_2$, i.e., where the edge
ends; and we say $col((u,v)) \equiv col((M[i_1,j_1],M[i_2,j_2])) =j_2$. 
Now we traverse all the edges and construct a set of sets
$\mathcal S = \{S_1\ldots S_\ell\}$ such that each $S_i, 1\leq i\leq
\ell$ contains the edges that are `same'. The set $S_i$ is named by
the (same) label of the edges it contains and we may refer to $S_i$
using its name. Now, we construct $P$-masks $P_{S_i}, 1\leq
i\leq \ell$ such that $P_{S_i}[k] = 1$ if, and only if, there is an
edge $(u,v)\in S_i$ having $col((u,v)) = k$. Clearly here, $\ell =
O(m)$.

Note that, to locate the appropriate $P$-mask, we again need to perform a look up in the database constructed during the construction of the $P$-masks. To maintain all $P$-masks we are keeping a $2$D array indexed by the consecutive vertices of an edge.

\begin{table}
\begin{center}
\begin{minipage}{\textwidth}
\begin{tabular}{|c|c|c|c|c|c|c|c|c|c|}
\hline
~ ~ ~ &~ $P_{(a,a)}$~ &~ $P_{(a,b)}$~ &~ $P_{(a,c)}$~ &~ $P_{(b,a)}$~ &~ $P_{(b,b)}$~ &~ $P_{(b,c)}$~ &~ $P_{(c,a)}$~ &~ $P_{(c,b)}$~ &~ $P_{(x,x)}$\footnote{Here, $P_{(x,x)}$ indicates the edges that are not present in the $\mathcal{P}$-graph.}\\
\hline
1&1&1&1&1&1&1&1&1&1\\
\hline
2&0&1&1&0&0&0&1&0&0\\
\hline
3&1&1&0&0&0&1&1&1&0\\
\hline
4&0&1&0&1&1&0&1&1&0\\
\hline
5&0&1&0&1&1&0&0&0&0\\
\hline
\end{tabular}
\end{minipage}
\caption{$P$-masks for patern $acbab$}
\label{newP_mask}
\end{center}
\end{table}

\begin{figure}[hbtp]
\begin{center}
\includegraphics[width=0.6\textwidth,angle=0]{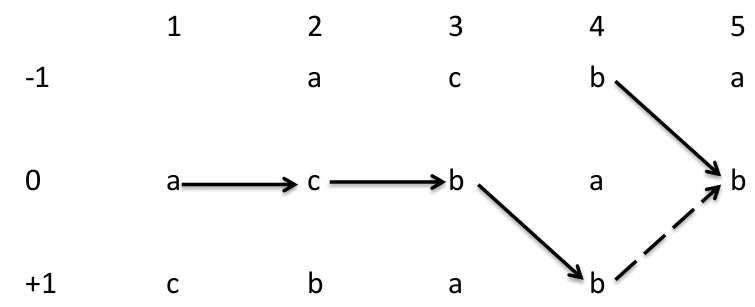}
\caption{Partial $\mathcal P$-graph for pattern $acbab$ } \label{swapproblem}
\end{center}
\end{figure}

However, the $P$-masks as defined above, and $D$-masks defined before are not sufficient to solve the SM problem as shown below with an example. Please note that the definition of $P$-mask in SMALGO-\rom{2} (i.e., the current algorithm) is different than that of SMALGO-\rom{1} (i.e., the algorithm presented in section~\ref{AIIR}).

\begin{example}\label{columnProblem}
In Table~\ref{newP_mask}, at the column named `P(b,b)', the value is $10011$ for pattern $acbab$. The `$1$' in the rightmost bit indicates that either one or both edge $((M[-1,4],M[0,5])$, $(M[+1,4],M[0,5]))$ exists, as shown in Figure~\ref{swapproblem}. We can not find out which one actually exists because our $P$-mask values are only dependent on the column positions ( i.e., the edge starts at Column $4$ and ends at Column $5$ ) irrespective of row positions ($-1, 0, +1$). So the algorithm will accept $acbbb$ as a swapped version
of the pattern $acbab$ which is clearly a false match.
\end{example} 

To solve the problem we need to be able to tell which \emph{level change} has occurred, \emph{Upward Change} or \emph{Middleward Change}. So, we introduce three new masks called $up$-$mask$, $down$-$mask$ and $middle$-$mask$ as discussed below. 

\begin{enumerate}
\item We construct \textbf{up-masks}, $up_{S_j}$, $1\leq j\leq \ell$ such that $up_{S_j}[g] = 1$ if and only if $edge(u_0,u_1) \equiv (M[+1,g-1],M[-1,g])$ exists with $col((u_0,u_1)) = g$.

\item We construct \textbf{down-masks}, $down_{S_j}$, $1\leq j\leq \ell$ such that $down_{S_j}[g] = 1$ if and only if either $edge(u_0,u_1) \equiv (M[-1,g-1],M[+1,g])$ or $edge(u_0,u_1) \equiv (M[0,g-1],M[+1,g])$ exists with $col((u_0,u_1)) = g$.

\item We construct \textbf{middle-masks}, $middle_{S_j}$, $1\leq j\leq \ell$ such that $middle_{S_j}[g] = 1$ if and only if either $edge(u_0,u_1) \equiv (M[0,g-1],M[0,g])$ or $edge(u_0,u_1) \equiv (M[-1,g-1],M[0,g])$ exists with $col((u_0,u_1)) = g$.

\end{enumerate}

The motivation and usefulness of the $3$ masks defined above will be clear from the following discussion.
It is easy to see that, to get a match, after a \emph{level change} at a particular position $(i_1,j)$, another \emph{level change} must occur at the next position, i.e., at position $(i_2,j+1)$ in the Matrix $M$; otherwise there can be no match. So we do the following based on the structure of the $\mathcal P$-graph.

\begin{enumerate}

\item If a \emph{Downward change} has occurred  then we have to check whether an \emph{Upward Change} occurs at the next position. We can do that by saving the previous \emph{down-mask} ($down_{(T_{j-1},T_{j})}$) and matching that value with the current \emph{up-mask} ($up_{(T_{j},T_{j+1})}$) and $R_j$. Otherwise there can be no match.

\item If an \emph{Upward Change} has occurred then we have to check whether \emph{Downward change} or a \emph{Middle-ward change} occurs at the next position. We can do that by saving the previous \emph{up-mask} ($up_{(T_{j-1},T_{j})}$) and matching that value with current \emph{down-mask}  ($down_{(T_{j},T_{j+1})}$), \emph{middle-mask} ($middle_{(T_{j},T_{j+1})}$) and $R_j$. Otherwise there can be no match.

\item This process continues repeatedly until either an \emph{End of Swap} occurs or an end of pattern is encountered. To check whether an end of swap occurs we have to keep previous \emph{up-mask} ($up_{(T_{j-1},T_{j})}$) and match that value with current \emph{middle-mask} ($middle_{(T_{j},T_{j+1})}$) and $R_j$. 
\end{enumerate}

In our algorithm, each of the previous checkings have to be done while we process each character. The algorithm is formally presented in Algorithm~\ref{Algorithm for Approximate String Matching Allowing for Fixed Length Translocation(Improved)}. The preprocessing of the algorithm is presented in Algorithm~\ref{algorithm for preprocessing}. In the algorithm, we are using a $2$-D array for $P$-masks, up-masks, down-masks and middle-masks. Example~\ref{Ex_match2} shows a complete execution of our algorithm.

\begin{example}\label{Ex_match2}
Suppose, $P = acbab$ and $T = acbbabcabab$. The $D$-masks of $P$ are shown in Table~\ref{common Dmask}. The $P$-masks of $P$ are shown in Table~\ref{newP_mask}. The new mask values are shown in Table~\ref{ournewmask} and Table~\ref{ourrightresult} shows the detail computation of $R$ bit array. Explanation of the terms used in Table~\ref{ourrightresult} are as follows:   

\begin{itemize}
\item[SH] Right Shift Operation on the previous column
\item[$D_x$] $D$-Mask value for character `x'
\item[$P_{(x,y)}$] $P$-Mask value of the set (x,y)
\item[$R_j$] Value of $R$ after $T_j$ has been proceed ( `1' in $m$-th row of $R_j$ column indicates that a match has been found )
\end{itemize}
\end{example}


\begin{table}
\begin{center}
\begin{tabular}{|c|c|c|c|c|c|c|c|c|}
\hline
~  ~&~ up-mask~ &~ middle-mask~ &~ down-mask~ \\
\hline
~ (a,a)~ &~ 00000~ &~ 00000~ &~ 00100~ \\
\hline
~ (a,b)~ &~ 00010~ &~ 00101~ &~ 01000~ \\
\hline
~ (a,c)~ &~ 00000~ &~ 01000~ &~ 10000~ \\
\hline
~ (b,a)~ &~ 00001~ &~ 00011~ &~ 00000~ \\
\hline 
~ (b,b)~ &~ 00000~ &~ 00001~ &~ 00010~ \\
\hline
~ (b,c)~ &~ 00100~ &~ 00000~ &~ 10000~ \\
\hline
~ (c,a)~ &~ 01000~ &~ 00010~ &~ 00100~ \\
\hline
~ (c,b)~ &~ 00000~ &~ 00100~ &~ 00010~ \\
\hline
\end{tabular}
\caption{Masks for Algorithm 2 of the pattern in Example~\ref{Ex_match2}} \label{ournewmask}
\end{center}
\end{table}


\begin{sidewaystable}
\scriptsize 
\tabcolsep 0.4pt
\begin{minipage}{\textwidth}
\begin{tabular}
{|c|c|c|c|c|c|c|c|c|c|c|c|c|c|c|c|c|c|c|c|c|c|c|c|c|c|c|c|c|c|c|c|c|c|c|c|c|c|c|c|c|c|c|c|c|c|}
\hline
~ ~ &~ ~ &SH&$D_a$&$P_{(x,x)}$\footnote{Here, $P_{(x,x)}$ indicates the edges that are not present in the $\mathcal{P}$-graph.} & $R_1$ &SH&$D_c$&$P_{(a,c)}$&$R_2$&SH&$D_b$&$P_{(c,b)}$&$R_3$&SH&$D_b$&$P_{(b,b)}$&$R_4$&SH&$D_a$&$P_{(b,a)}$&$R_5$&SH&$D_b$&$P_{(a,b)}$&$R_6$&SH&$D_c$&
$P_{(b,c)}$&$R_7$&SH&$D_a$&$P_{(c,a)}$&$R_8$&SH&$D_b$&$P_{(a,b)}$&$R_9$&SH&$D_a$&$P_{(b,a)}$&$R_{10}$&SH&$D_b$&$P_{(a,b)}$&$R_{11}$\\
\hline
1&0&1&1&1&1&1&1&1&1&1&0&1&0&1&0&1&0&1&1&1&1&1&0&1&0&1&1&1&1&1&1&1&1&1&0&1&0&1&1&1&1&1&0&1&0\\
\hline
2&0&0&1&0&0&1&1&1&1&1&1&0&0&0&1&0&0&0&1&0&0&1&1&1&1&0&1&0&0&1&1&1&1&1&1&1&1&0&1&0&0&1&1&1&1\\
\hline
3&0&0&1&0&0&0&1&0&0&0&1&1&1&0&1&0&0&0&1&0&0&0&1&1&0&1&1&1&1&0&1&1&0&1&1&1&1&1&1&0&0&0&1&1&0\\
\hline
4&0&0&1&0&0&0&0&0&0&0&1&1&0&1&1&1&1&0&1&1&0&0&1&1&0&0&0&0&0&1&1&1&1&0&1&1&0&1&1&1&1&0&1&1&0\\
\hline 
5&0&0&1&0&0&0&0&0&0&0&1&0&0&0&1&1&0&1&1&1&1&0&1&1&0&0&0&0&0&0&1&0&0&1&1&1&1&0&1&1&0&1&1&1&1\\
\hline
\end{tabular}
\end{minipage}
\caption{Detailed Calculation for text in Example~\ref{Ex_match2}} \label{ourrightresult}
\end{sidewaystable}

\begin{algorithm}
\caption{Algorithm for Computation of all the Masks [Preprocessing] for SMALGO-\rom{2}}
\label{algorithm for preprocessing}
\begin{algorithmic}[1]
\REQUIRE{\textbf{Pattern p}}
\STATE{$x \leftarrow pattern_{length}$}
\FOR{$i=0$ to $x-2$}\STATE{$pmask[p[i]][p[i+1]] \leftarrow pmask[p[i]][p[i+1]] | (1<<(x-i-2))$}
\STATE{{$middle[p[i]][p[i+1]] \leftarrow middle[p[i]][p[i+1]] | (1<<(x-i-2))$}}
\ENDFOR
\FOR{$i=0$ to $x-3$}\STATE{$pmask[p[i]][p[i+2]] \leftarrow pmask[p[i]][p[i+2]] | (1<<(x-i-2)) | (1<<(x-i-3))$}
\STATE{{$down[p[i]][p[i+2]] \leftarrow down[p[i]][p[i+2]] | (1<<(x-i-2))$}}
\STATE{{$middle[p[i]][p[i+2]] \leftarrow middle[p[i]][p[i+2]] | (1<<(x-i-3))$}}
\ENDFOR
\FOR{$i=1$ to $x-1$}\STATE{$pmask[p[i]][p[i-1]] \leftarrow pmask[p[i]][p[i-1]] | (1<<(x-i-1))$}
\STATE{{$up[p[i]][p[i-1]] \leftarrow up[p[i]][p[i-1]] | (1<<(x-i-1))$}}
\ENDFOR
\FOR{$i=1$ to $x-4$}\STATE{$pmask[p[i]][p[i+3]] \leftarrow pmask[p[i]][p[i+3]] | (1<<(x-i-3))$}
\STATE{{$down[p[i]][p[i+3]] \leftarrow up[p[i]][p[i+3]] | (1<<(x-i-3))$}}
\ENDFOR
\FOR{$i=1$ to $x-1$}\STATE{$d[p[i]] \leftarrow d[p[i]] | (1<<(x-i)) | (1<<(x-i-1)) | (1<<(x-i-2))$}
\ENDFOR
\RETURN{\textbf{$P$-masks pmask[]}, \textbf{$D$-masks d[]}, \textbf{up-masks up[]}, \textbf{down-masks down[]} and \textbf{middle-masks middle[]} for Pattern p}
\end{algorithmic}
\end{algorithm}

\begin{algorithm}
\caption{SMALGO-\rom{2}}
\label{Algorithm for Approximate String Matching Allowing for Fixed Length Translocation(Improved)}
\begin{algorithmic}[1]
\REQUIRE{\textbf{Text T}, \textbf{up-mask up}, \textbf{down-mask down}, \textbf{middle-mask middle}, \textbf{P-mask pmask},   \textbf{D-mask D}  for given pattern p}
\STATE{$R_{0} \leftarrow 2^{pattern_{length}}-1$}
\STATE{$checkup \leftarrow checkdown \leftarrow 0$}
\STATE{$R_{0} \leftarrow R_{0}$ \& $D_{T_{0}}$}
\STATE{$R_{1} \leftarrow R_{0} >> 1$}
\STATE{$x \leftarrow 1$}
\FOR{$j=0$ to $(n-2)$}
	\STATE{$R_{j} \leftarrow R_{j}$ \& $pmask_{(T_{j},T_{j+1})}$ \& $D_{T_{j+1}}$}
	\STATE{$temp \leftarrow prevcheckup >> 1$}
	\STATE{$checkup \leftarrow checkup$ $|$  $up_{(T_{j},T_{j+1})}$}
	\STATE{$checkup \leftarrow checkup$ \& $\sim down_{(T_{j},T_{j+1})}$ \& $\sim middle_{(T_{j},T_{j+1})}$}
	\STATE{$prevcheckup \leftarrow checkup$}
	\STATE{$R_{j} \leftarrow \sim (temp$ \& $checkup)$ \& $R_{j}$}
	\STATE{$temp \leftarrow prevcheckdown >> 1$}
	\STATE{$checkdown \leftarrow checkdown$ $|$  $down_{(T_{j},T_{j+1})}$}
	\STATE{$checkdown \leftarrow checkdown$ \& $\sim up_{(T_{j},T_{j+1})}$}
	\STATE{$prevcheckdown \leftarrow checkdown$}
	\STATE{$R_{j} \leftarrow \sim (temp$ \& $checkdown)$ \& $R_{j}$}
	\IF{$(R_{j}$ \& $x) = x$} \STATE{Match found ending at position $(j-1)$}
	\ENDIF
	\STATE{$R_{j+1} \leftarrow R_{j} >> 1$}
	\STATE{$checkup \leftarrow checkup >> 1$}
	\STATE{$checkdown \leftarrow checkdown >> 1$}
\ENDFOR
\end{algorithmic}
\end{algorithm}

\subsection{Analysis of SMALGO-\rom{2}}\label{analysis}

\begin{table}
\begin{center}
\begin{tabular}{|c|c|}
\hline
Phase&Running Time\\
\hline
Computation of $D$-masks&$O(m/w(m+|\Sigma|))$\\
\hline
Computation of $P$-masks&$O(m/w (m+{\vert\Sigma\vert}^{2}))$\\
\hline
Computation of $Up$-masks&$O(m/w (m+{\vert\Sigma\vert}^{2}))$\\
\hline
Computation of $Down$-masks&$O(m/w (m+{\vert\Sigma\vert}^{2}))$\\
\hline
Computation of $Middle$-masks&$O(m/w (m+{\vert\Sigma\vert}^{2}))$\\
\hline
Running time of Algorithm~\ref{Algorithm for Approximate String Matching Allowing for Fixed Length Translocation(Improved)}&$O(m/w ~n)$\\
\hline
\end{tabular}
\caption{Running times of the different phases of of SMALGO-\rom{2}}
\label{Table_RunningTime2}
\end{center}
\end{table}

The running times of the different phases of SMALGO-\rom{2} is listed in Table~\ref{Table_RunningTime2}. In SMALGO-\rom{2}, we first initialize all the entries of $P$-masks which requires $O(m/w {\vert\Sigma\vert}^{2})$ time. Then, we start traversing the edges and corresponding $P$-masks in a name database ($2$-D array). Finding and updating the $P$-mask of corresponding edges can be done in constant time. As we have $O(m)$ edges, the total time needed for computation of $P$-mask is $O(m/w (m + {\vert\Sigma\vert}^{2}))$. 
Similarly, the computation of up-masks, down-masks and middle-masks can be done in $O(m/w (m + {\vert\Sigma\vert}^{2})$ time as well. The computation of $D$-mask takes $O(m/w (m + {\vert\Sigma\vert}))$ time.
So the whole preprocessing takes $O(m/w (4(m + {\vert\Sigma\vert}^{2}) + m + {\vert\Sigma\vert}))))$ time. Assuming constant alphabet $\Sigma$ and pattern of size compatible with machine word length the preprocessing time becomes $O(m)$.

With all the masks at our hand, for our problem, we simply need to compute $R_{j}$ by some simple calculation. Each step of the calculation, including locating the appropriate masks, needs constant amount of time. So, in total the construction of $R$ values require $O(m/w ~n)$ which is $O(n)$ when $w \sim m$.

Therefore, in total the running time for SMALGO-\rom{2}, is linear
assuming a constant alphabet and a pattern size similar to the word size of the target machine.

\section{Experimental Results}\label{experiment}

We have conducted extensive experiments to compare the actual running time of the existing (non FFT) swap matching algorithms in the literature~\cite{CS,CCS} with ours. In this section, we present our findings based on the experiments conducted. The following acronyms are used in the presented results to identify different algorithms.

\begin{itemize}
\item[CS] CROSS-SAMPLING algorithm of~\cite{CS}
\item[BPCS] BP-CROSS-SAMPLING algorithm of~\cite{CS}
\item[BCS] BACKWARD-CROSS-SAMPLING algorithm of~\cite{CCS}
\item[BPBCS] BP-BACKWARD-CROSS-SAMPLING algorithm of~\cite{CCS}
\item[ALG-\rom{1}] SMALGO-\rom{1} of this paper
\item[ALG-\rom{2}] SMALGO-\rom{2} of this paper
\end{itemize}

We have chosen to exclude the naive algorithm and all algorithms in the literature based on FFT techniques from our experiments, because, the overhead of such algorithms is quite high resulting in a bad performance. All algorithms have been implemented in Microsoft Visual C++ in Release Mode on a PC with Intel Pentium D processor of 2.8 GHz having a memory of 2GB.

\subsection{Datasets}
All algorithms have been tested on random texts, on a Genome sequence, on a Protein sequence and on a natural language text buffer with patterns of length, $m$ = 4, 8, 12, 16, 20, 24, 28, 32. In the Tables below running times have been expressed in the hundredth of a second and best results are highlighted.

In the case of random texts we have adopted a similar strategy of~\cite{CS,CCS}. In particular, the algorithm has been tested on six $Rand\Sigma$ problem sets (for $\vert\Sigma\vert$ = 4, 8, 16, 32, 64 and 128). Each $Rand\Sigma$ problem consists in searching a set of 100 random patterns for any given length value in a 4MB long random text over a common alphabet of size $\vert\Sigma\vert$. In order to make the test more effective, in our experiments half of the patterns are randomly chosen and rests are picked from the text randomly so that they surely appear in the text at least once.

We also follow a strategy similar to that of~\cite{CS,CCS} for the tests on real world problems. 
We have been performed tests on a Genome sequence, on a protein sequence and on a natural text buffer. The genome sequence we used for the tests is a sequence of 4,638,690 base pairs of $Escherichia~ coli$ taken from the file \emph{E.coli} of the large Canterbury Corpus~\cite{corpus}.
The tests on the Protein sequence have been performed using a 2.4MB file containing a protein sequence from the Human Genome with 22 different characters. The experiments on the natural language text buffer have been done on the file \emph{world192.txt} (The CIA World Fact Book) of the Large Canterbury Corpus~\cite{corpus}. This file contains 2,473,400 characters drawn from an alphabet of 93 different characters.

\subsection{Running Times of Random Problems}

The running times for different algorithms for this experiment are reported in Tables~\ref{Rand4} -~\ref{Rand128}. From the results, we see that, in general, ALG-\rom{1} (SMALGO-\rom{1}) runs faster than BPBCS for smaller patterns and small alphabet size whereas BPBCS performs better when pattern and alphabet size are relatively large.


\begin{table}
\begin{center}
\begin{tabular}{|c|c|c|c|c|c|c|c|c|}
\hline
m&4&8&12&16&20&24&28&32\\
\hline
CS &~ 61.247~ &~ 61.137~ &~ 61.252~ &~ 61.043~ &~ 61.468~ &~ 63.472~ &~ 68.489~ &~ 66.090~ \\
\hline
BCS &~ 33.366~ &~ 22.865~ &~ 18.523~ &~ 16.580~ &~ 15.289~ &~ 14.585~ &~ 13.628~ &~ 13.087~ \\
\hline
BPCS &~ 1.914~ &~ 1.867~ &~ 1.849~ &~ 1.864~ &~ 1.861~ &~ 1.908~ &~ 1.859~ &~ 1.860~ \\
\hline
~ BPBCS~ &~ 3.552~ &~ 2.001~ &~ 1.451~ &~ 1.124~ &~ 0.968~ &~ 0.835~ &~ 0.737~ &~ 0.672~ \\
\hline 
ALG-\rom{2} &~ 4.062~ &~ 4.081~ &~ 4.090~ &~ 4.093~ &~ 4.124~ &~ 4.082~ &~ 4.085~ &~ 4.095~ \\
\hline
ALG-\rom{1} &~ \textbf{0.631}~ &~ \textbf{0.626}~ &~ \textbf{0.631}~ &~ \textbf{0.631}~ &~ \textbf{0.636}~ &~ \textbf{0.636}~ &~ \textbf{0.640}~ &~ \textbf{0.629}~ \\
\hline
\end{tabular}
\caption{Running time for Rand4 problems}
\label{Rand4}
\end{center}
\end{table}

\begin{table}
\begin{center}
\begin{tabular}{|c|c|c|c|c|c|c|c|c|}
\hline
m&4&8&12&16&20&24&28&32\\
\hline
CS &~ 52.447~ &~ 52.407~ &~ 52.356~ &~ 52.357~ &~ 52.457~ &~ 52.424~ &~ 52.443~ &~ 52.424~ \\
\hline
BCS &~ 23.139~ &~ 16.532~ &~ 12.539~ &~ 10.810~ &~ 9.701~ &~ 9.111~ &~ 8.502~ &~ 7.996~ \\
\hline
BPCS &~ 1.861~ &~ 1.858~ &~ 1.857~ &~ 1.864~ &~ 1.896~ &~ 1.905~ &~ 2.001~ &~ 1.853~ \\
\hline
~ BPBCS~ &~ 2.156~ &~ 1.319~ &~ 0.944~ &~ 0.743~ &~ \textbf{0.621}~ &~ \textbf{0.533}~ &~ \textbf{0.476}~ &~ \textbf{0.421}~ \\
\hline 
ALG-\rom{2} &~ 4.061~ &~ 4.066~ &~ 4.051~ &~ 4.062~ &~ 4.060~ &~ 4.059~ &~ 4.064~ &~ 4.063~ \\
\hline
ALG-\rom{1} &~ \textbf{0.674}~ &~ \textbf{0.684}~ &~ \textbf{0.671}~ &~ \textbf{0.675}~ &~ 0.654~ &~ 0.676~ &~ 0.633~ &~ 0.633~ \\
\hline
\end{tabular}
\caption{Running time for Rand8 problems}
\label{Rand8}
\end{center}
\end{table}

\begin{table}
\begin{center}
\begin{tabular}{|c|c|c|c|c|c|c|c|c|}
\hline
m&4&8&12&16&20&24&28&32\\
\hline
CS &~ 51.632~ &~ 52.622~ &~ 53.139~ &~ 53.450~ &~ 51.829~ &~ 55.243~ &~ 53.136~ &~ 52.847~ \\
\hline
BCS &~ 18.718~ &~ 13.350~ &~ 11.130~ &~ 9.084~ &~ 7.711~ &~ 7.045~ &~ 6.578~ &~ 6.103~ \\
\hline
BPCS &~ 1.864~ &~ 1.870~ &~ 1.857~ &~ 1.844~ &~ 1.862~ &~ 1.851~ &~ 1.856~ &~ 1.864~ \\
\hline
~ BPBCS~ &~ 1.340~ &~ 0.917~ &~ 0.700~ &~ \textbf{0.555}~ &~ \textbf{0.459}~ &~ \textbf{0.393}~ &~ \textbf{0.349}~ &~ \textbf{0.314}~ \\
\hline 
ALG-\rom{2} &~ 4.063~ &~ 4.062~ &~ 4.066~ &~ 4.063~ &~ 4.075~ &~ 4.083~ &~ 4.062~ &~ 4.074~ \\
\hline
ALG-\rom{1} &~ \textbf{0.636}~ &~ \textbf{0.634}~ &~ \textbf{0.662}~ &~ 0.631~ &~ 0.635~ &~ 0.640~ &~ 0.632~ &~ 0.643~ \\
\hline
\end{tabular}
\caption{Running time for Rand16 problems}
\label{Rand16}
\end{center}
\end{table}

\begin{table}
\begin{center}
\begin{tabular}{|c|c|c|c|c|c|c|c|c|}
\hline
m&4&8&12&16&20&24&28&32\\
\hline
CS &~ 52.644~ &~ 54.104~ &~ 51.123~ &~ 51.296~ &~ 50.795~ &~ 53.045~ &~ 54.296~ &~ 53.458~ \\
\hline
BCS &~ 16.628~ &~ 11.472~ &~ 8.682~ &~ 7.618~ &~ 6.684~ &~ 5.966~ &~ 5.474~ &~ 5.318~ \\
\hline
BPCS &~ 1.862~ &~ 1.863~ &~ 1.866~ &~ 1.858~ &~ 1.862~ &~ 1.851~ &~ 1.910~ &~ 1.856~ \\
\hline
~ BPBCS~ &~ 0.950~ &~ 0.637~ &~ \textbf{0.532}~ &~ \textbf{0.453}~ &~ \textbf{0.394}~ &~ \textbf{0.363}~ &~ \textbf{0.307}~ &~ \textbf{0.274}~ \\
\hline 
ALG-\rom{2} &~ 4.100~ &~ 4.094~ &~ 4.104~ &~ 4.093~ &~ 4.102~ &~ 4.108~ &~ 4.099~ &~ 4.101~ \\
\hline
ALG-\rom{1} &~ \textbf{0.636}~ &~ \textbf{0.633}~ &~ 0.630~ &~ 0.632~ &~ 0.629~ &~ 0.631~ &~ 0.631~ &~ 0.629~ \\
\hline
\end{tabular}
\caption{Running time for Rand32 problems}
\label{Rand32}
\end{center}
\end{table}

\begin{table}
\begin{center}
\begin{tabular}{|c|c|c|c|c|c|c|c|c|}
\hline
m&4&8&12&16&20&24&28&32\\
\hline
CS &~ 49.482~ &~ 47.775~ &~ 50.448~ &~ 49.668~ &~ 52.965~ &~ 51.000~ &~ 52.663~ &~ 53.103~ \\
\hline
BCS &~ 14.892~ &~ 9.850~ &~ 7.615~ &~ 6.481~ &~ 6.692~ &~ 5.484~ &~ 5.242~ &~ 5.066~ \\
\hline
BPCS &~ 1.846~ &~ 1.855~ &~ 1.864~ &~ 1.919~ &~ 1.923~ &~ 1.851~ &~ 1.917~ &~ 1.863~ \\
\hline
~ BPBCS~ &~ 0.739~ &~ \textbf{0.475}~ &~ \textbf{0.370}~ &~ \textbf{0.334}~ &~ \textbf{0.294}~ &~ \textbf{0.282}~ &~ \textbf{0.271}~ &~ \textbf{0.233}~ \\
\hline
ALG-\rom{2} &~ 4.334~ &~ 4.341~ &~ 4.342~ &~ 4.336~ &~ 4.347~ &~ 4.355~ &~ 4.390~ &~ 4.341~ \\
\hline
ALG-\rom{1} &~ \textbf{0.635}~ &~ 0.635~ &~ 0.686~ &~ 0.626~ &~ 0.637~ &~ 0.636~ &~ 0.632~ &~ 0.630~ \\
\hline
\end{tabular}
\caption{Running time for Rand64 problems}
\label{Rand64}
\end{center}
\end{table}

\begin{table}
\begin{center}
\begin{tabular}{|c|c|c|c|c|c|c|c|c|}
\hline
m&4&8&12&16&20&24&28&32\\
\hline
CS &~ 49.939~ &~ 48.608~ &~ 50.570~ &~ 52.593~ &~ 50.797~ &~ 50.075~ &~ 49.481~ &~ 49.640~ \\
\hline
BCS &~ 14.411~ &~ 9.541~ &~ 7.513~ &~ 6.546~ &~ 6.769~ &~ 5.242~ &~ 4.743~ &~ 4.573~ \\
\hline
BPCS &~ 1.855~ &~ 1.866~ &~ 1.849~ &~ 1.851~ &~ 1.855~ &~ 1.858~ &~ 1.861~ &~ 1.864~ \\
\hline
~ BPBCS~ &~ 0.688~ &~ \textbf{0.430}~ &~ \textbf{0.325}~ &~ \textbf{0.288}~ &~ \textbf{0.274}~ &~ \textbf{0.229}~ &~ \textbf{0.216}~ &~ \textbf{0.204}~ \\
\hline
ALG-\rom{2} &~ 4.429~ &~ 4.431~ &~ 4.426~ &~ 4.477~ &~ 4.408~ &~ 4.032~ &~ 4.422~ &~ 4.422~ \\
\hline
ALG-\rom{1} &~ \textbf{0.651}~ &~ 0.661~ &~ 0.661~ &~ 0.640~ &~ 0.653~ &~ 0.638~ &~ 0.633~ &~ 0.644~ \\
\hline
\end{tabular}
\caption{Running time for Rand128 problems}
\label{Rand128}
\end{center}
\end{table}

\begin{table}
\begin{center}
\begin{tabular}{|c|c|c|c|c|c|c|c|c|}
\hline
m&4&8&12&16&20&24&28&32\\
\hline
CS &~ 74.771~ &~ 74.729~ &~ 74.890~ &~ 74.564~ &~ 73.251~ &~ 77.605~ &~ 71.991~ &~ 73.472~ \\
\hline
BCS &~ 78.499~ &~ 52.828~ &~ 43.492~ &~ 38.730~ &~ 35.401~ &~ 33.078~ &~ 31.815~ &~ 30.234~ \\
\hline
BPCS &~ 2.224~ &~ 2.157~ &~ 2.152~ &~ 2.152~ &~ 2.146~ &~ 2.152~ &~ 2.159~ &~ 2.164~ \\
\hline
~ BPBCS~ &~ 3.977~ &~ 2.228~ &~ 1.619~ &~ 1.283~ &~ 1.084~ &~ 0.930~ &~ 0.830~ &~ 0.749~ \\
\hline
ALG-\rom{2} &~ 4.732~ &~ 4.734~ &~ 4.739~ &~ 4.744~ &~ 4.745~ &~ 4.730~ &~ 4.709~ &~ 4.711~ \\
\hline
ALG-\rom{1} &~ \textbf{0.600}~ &~ \textbf{0.619}~ &~ \textbf{0.639}~ &~ \textbf{0.596}~ &~ \textbf{0.611}~ &~ \textbf{0.606}~ &~ \textbf{0.582}~ &~ \textbf{0.583}~ \\
\hline
\end{tabular}
\caption{Running time for a genome sequence ($\Sigma$ = $4$)} \label{genome}
\end{center}
\end{table}

\begin{table}
\begin{center}
\begin{tabular}{|c|c|c|c|c|c|c|c|c|}
\hline
m&4&8&12&16&20&24&28&32\\
\hline
CS &~ 55.764~ &~ 54.316~ &~ 55.594~ &~ 51.737~ &~ 50.797~ &~ 50.491~ &~ 50.154~ &~ 50.889~ \\
\hline
BCS &~ 33.825~ &~ 24.180~ &~ 21.250~ &~ 17.179~ &~ 14.344~ &~ 13.464~ &~ 12.506~ &~ 12.075~ \\
\hline
BPCS &~ 1.852~ &~ 1.853~ &~ 1.866~ &~ 1.866~ &~ 1.878~ &~ 1.859~ &~ 1.859~ &~ 2.027~ \\
\hline
~ BPBCS~ &~ 1.107~ &~ 0.771~ &~ 0.617~ &~ \textbf{0.504}~ &~ \textbf{0.428}~ &~ \textbf{0.367}~ &~ \textbf{0.314}~ &~ \textbf{0.282}~ \\
\hline
ALG-\rom{2} &~ 4.064~ &~ 4.086~ &~ 4.098~ &~ 4.062~ &~ 4.068~ &~ 4.056~ &~ 4.076~ &~ 4.075~ \\
\hline
ALG-\rom{1} &~ \textbf{0.613}~ &~ \textbf{0.606}~ &~ \textbf{0.599}~ &~ 0.611~ &~ 0.618~ &~ 0.602~ &~ 0.613~ &~ 0.622~ \\
\hline
\end{tabular}
\caption{Running time for a protein sequence ($\Sigma$ = $22$)} \label{protein}
\end{center}
\end{table}

\begin{table}
\begin{center}
\begin{tabular}{|c|c|c|c|c|c|c|c|c|}
\hline
m&4&8&12&16&20&24&28&32\\
\hline
CS &~ 30.205~ &~ 33.177~ &~ 30.626~ &~ 33.344~ &~ 28.872~ &~ 31.382~ &~ 30.744~ &~ 29.362~ \\
\hline
BCS &~ 29.917~ &~ 26.244~ &~ 25.232~ &~ 28.203~ &~ 30.407~ &~ 26.468~ &~ 30.966~ &~ 26.778~ \\
\hline
BPCS &~ 1.178~ &~ 1.166~ &~ 1.181~ &~ 1.154~ &~ 1.146~ &~ 1.150~ &~ 1.146~ &~ 1.132~ \\
\hline
~ BPBCS~ &~ 0.793~ &~ 0.701~ &~ .657~ &~ 0.723~ &~ 0.768~ &~ 0.691~ &~ 0.738~ &~ 0.694~ \\
\hline
ALG-\rom{2} &~ 2.508~ &~ 2.510~ &~ 2.506~ &~ 2.507~ &~ 2.546~ &~ 2.501~ &~ 2.510~ &~ 2.505~ \\
\hline
ALG-\rom{1} &~ \textbf{0.632}~ &~ \textbf{0.620}~ &~ \textbf{0.625}~ &~ \textbf{0.616}~ &~ \textbf{0.609}~ &~ \textbf{0.634}~ &~ \textbf{0.618}~ &~ \textbf{0.608}~ \\
\hline
\end{tabular}
\caption{Running time for a natural language text buffer ($\Sigma$ = $93$)} \label{natural}
\end{center}
\end{table}


\subsection{Running Times for Real World Problems}

The running time of different algorithms in these different experiments are reported in Tables~\ref{genome} -~\ref{natural}. From the experiments, we see that ALG-\rom{1} runs faster for all pattern lengths in genome sequence and natural language text buffer. However in protein sequence, ALG-\rom{1} performs best for smaller patterns whereas BPBCS performs better for larger patterns.




\section{\label{conclusion}Conclusion}
In this paper, we have revisited the Swap Matching problem, a well-studied variant of the classic pattern matching problem. In particular, we have presented a graph theoretic model to solve the swap matching problem and devised two novel algorithms based on this model. The resulting algorithms are adaptations of the classic shift-and algorithm~\cite{CharrasL04} and runs in linear time for finite alphabet if the pattern-length is similar to the word-size in the target machine. Note that our algorithms like the work of~\cite{CS,CCS} does not use FFT techniques. Though both algorithms are based on the same classic shift-and algorithm, they are different in their technique/approach. Moreover, the techniques used in our algorithms are quite simple and easy to implement as well as understand.

We believe that our graph theoretic model could be used to devise more efficient algorithms and a similar approach can be taken to model similar other variants of the classic pattern matching problem. Furthermore, it would be interesting to `swap' the definitions of $\mathcal T$- graph and $\mathcal P$- graph and investigate whether efficient pattern matching techniques for Directed Acyclic Graphs can be employed to devise efficient off-line and online algorithms for swap matching.

\bibliography{swap}
\bibliographystyle{abbrv}

\end{document}